\documentclass[12pt]{amsart}
\usepackage{amscd,amsthm,amsfonts,amssymb,amsmath,euscript}
\usepackage[dvips]{graphicx}
\usepackage[dvips]{graphics}
\usepackage[T2A]{fontenc}
\usepackage[cp866]{inputenc}

\newtheorem{theorem}{Theorem}[section]
\newtheorem{lemma}{Lemma}[section]

\newtheorem{example}{Example}[section]

\newtheorem{remark}{Remark}[section]

\newcommand{\Aut}{\mathop{\sf Aut}\nolimits}

\newcommand{\tr}{\mathop{\sf tr}\nolimits}
\newcommand{\End}{\mathop{\sf End}\nolimits}

\,

\newcommand{\ododlo}{\begin{picture}(20,10)(-2,-3)
\put(0,0){\line(1,0){16}}
\put(0,0){\circle*{3}}\put(16,0){\circle*{3}}
\end{picture}}\,

\,

\newcommand{\dvadva}{\begin{picture}(20,14)(-2,-3)
\put(0,-6){\line(1,0){16}}
\put(0,6){\line(1,0){16}}
\put(0,-6){\circle*{3}}\put(16,-6){\circle*{3}}
\put(0,6){\circle*{3}}\put(16,6){\circle*{3}}
\end{picture}}\,

\newcommand{\oddva}{\begin{picture}(20,14)(-2,-3)
\qbezier(0,0)(8,3)(16,6)
\qbezier(0,0)(8,-3)(16,-6)
\put(16,-6){\circle*{3}}
\put(0,0){\circle*{3}}
\put(16,6){\circle*{3}}
\end{picture}}\,

\newcommand{\dvaod}{\begin{picture}(20,14)(-2,-3)
\qbezier(0,6)(8,3)(16,0)
\qbezier(0,-6)(8,-3)(16,0)
\put(0,-6){\circle*{3}}
\put(16,0){\circle*{3}}
\put(0,6){\circle*{3}}
\end{picture}}\,

\newcommand{\ododld}{\begin{picture}(20,14)(-2,-3)
\qbezier(0,0)(8,12)(16,0)
\qbezier(0,0)(8,-12)(16,0)
\put(0,0){\circle*{3}}\put(16,0){\circle*{3}}
\end{picture}}\,

\,

\,

\newcommand{\tritri}{\begin{picture}(20,14)(-2,-3)
\put(0,-6){\line(1,0){16}}
\put(0,0){\line(1,0){16}}
\put(0,6){\line(1,0){16}}
\put(0,-6){\circle*{3}}\put(16,-6){\circle*{3}}
\put(0,0){\circle*{3}}\put(16,0){\circle*{3}}
\put(0,6){\circle*{3}}\put(16,6){\circle*{3}}
\end{picture}}\,

\newcommand{\odtri}{\begin{picture}(20,14)(-2,-3)
\qbezier(0,0)(8,3)(16,6)
\put(0,0){\line(1,0){16}}
\qbezier(0,0)(8,-3)(16,-6)
\put(16,-6){\circle*{3}}
\put(0,0){\circle*{3}}\put(16,0){\circle*{3}}
\put(16,6){\circle*{3}}
\end{picture}}\,

\newcommand{\triod}{\begin{picture}(20,14)(-2,-3)
\qbezier(0,6)(8,3)(16,0)
\put(0,0){\line(1,0){16}}
\qbezier(0,-6)(8,-3)(16,0)
\put(0,-6){\circle*{3}}
\put(0,0){\circle*{3}}\put(16,0){\circle*{3}}
\put(0,6){\circle*{3}}
\end{picture}}\,

\newcommand{\dvatri}{\begin{picture}(20,14)(-2,-3)
\qbezier(0,3)(8,4.5)(16,6)
\put(0,-6){\line(1,0){16}}
\qbezier(0,3)(8,1.5)(16,0)
\put(0,-6){\circle*{3}}\put(16,-6){\circle*{3}}
\put(0,3){\circle*{3}}\put(16,0){\circle*{3}}
\put(16,6){\circle*{3}}
\end{picture}}\,

\newcommand{\tridva}{\begin{picture}(20,14)(-2,-3)
\qbezier(0,6)(8,4.5)(16,3)
\put(0,-6){\line(1,0){16}}
\qbezier(0,0)(8,1.5)(16,3)
\put(0,-6){\circle*{3}}\put(16,-6){\circle*{3}}
\put(0,0){\circle*{3}}\put(16,3){\circle*{3}}
\put(0,6){\circle*{3}}
\end{picture}}\,

\newcommand{\oddvalt}{\begin{picture}(20,14)(-2,-3)
\qbezier(0,0)(8,8)(16,6)
\qbezier(0,0)(8,-2)(16,6)
\qbezier(0,0)(8,-3)(16,-6)
\put(16,-6){\circle*{3}}
\put(0,0){\circle*{3}}
\put(16,6){\circle*{3}}
\end{picture}}\,

\newcommand{\dvadvalt}{\begin{picture}(20,14)(-2,-3)
\qbezier(0,3)(8,10)(16,3)
\qbezier(0,3)(8,-4)(16,3)
\put(0,-6){\line(1,0){16}}
\put(0,-6){\circle*{3}} \put(16,-6){\circle*{3}}
\put(16,3){\circle*{3}}
\put(0,3){\circle*{3}}
\end{picture}}\,

\newcommand{\zigzag}{\begin{picture}(20,14)(-2,-3)
\qbezier(0,-6)(8,-4)(16,-2)
\qbezier(0,2)(8,0)(16,-2)
\qbezier(0,2)(8,4)(16,6)
\put(0,-6){\circle*{3}}
\put(16,-2){\circle*{3}}
\put(0,2){\circle*{3}}
\put(16,6){\circle*{3}}
\end{picture}}\,

\newcommand{\dvaodlt}{\begin{picture}(20,14)(-2,-3)
\qbezier(0,6)(8,8)(16,0)
\qbezier(0,6)(8,-2)(16,0)
\qbezier(0,-6)(8,-3)(16,0)
\put(0,-6){\circle*{3}}
\put(16,0){\circle*{3}}
\put(0,6){\circle*{3}}
\end{picture}}\,

\newcommand{\ododlt}{\begin{picture}(20,14)(-2,-3)
\qbezier(0,0)(8,12)(16,0)
\put(0,0){\line(1,0){16}}
\qbezier(0,0)(8,-12)(16,0)
\put(0,0){\circle*{3}}\put(16,0){\circle*{3}}
\end{picture}}\,

\,

\,

\,

\,

\,

\,

\,

\,

\,

\,

\,

\,

\,

\,

\,

\,

\,

\,

\,

\,

\,

\,

\,

\,

\,

\,

\,

\,

\,

\,

\,

\,

\,

\,

\,

\,

\def\CF{Cardy-Frobenius }

\begin{document}

\begin{center}

\hfill FIAN/TD-27/12\\
\hfill ITEP/TH-49/12

\vspace{1cm}

\end{center}

\centerline{\large\bf Cardy-Frobenius extension of algebra}

\vspace{1ex}

\centerline{\large\bf of cut-and-join operators}

\vspace{1cm}

\centerline {A.Mironov, A.Morozov, S.Natanzon}
\address{Theory Department, Lebedev Physical Institute, Moscow, Russia;
Institute for Theoretical and Experimental Physics, Moscow, Russia}
\email{mironov@itep.ru; mironov@lpi.ru}
\address{Institute for Theoretical and Experimental Physics}
\email{morozov@itep.ru}
\address {Department of Mathematics, Higher School of Economics, Moscow, Russia;
A.N.Belozersky Institute, Moscow State University, Russia;
Institute for Theoretical and Experimental Physics}
\email{natanzons@mail.ru}

\vspace{1cm}

\centerline{ABSTRACT}

\bigskip

{\footnotesize Motivated by the algebraic open-closed string models,
we introduce and discuss an infinite-dimensional
counterpart of the open-closed Hurwitz theory describing branching coverings generated both by
the compact oriented surfaces and by the foam surfaces.
We manifestly construct the corresponding infinite-dimensional equipped \CF algebra, with the
closed and open sectors are represented by conjugation classes of
permutations and the pairs of permutations, i.e. by the algebra of
Young diagrams and bipartite graphes respectively. }


\vspace{1cm}

\section {Main result}

\subsection {\CF algebra}
We start with reminding the definition of the (finite-dimensional) equipped \CF algebra
following \cite{AN} and \cite{AN2}.

\vspace{1ex}

\textit{Frobenius pair} is a set $(C,l^C)$ which consists of a finite-dimensional associative algebra $C$ with an
identity element and a linear functional  $l^C:C\rightarrow\mathbb{C}$ such that the bilinear form
$(c_1,c_2)_C= l^C(c_1c_2)$ which it generates is non-degenerated.

\textit{Casimir} of the Frobenius pair $(C,l^C)$ is the element
$K_C=\sum\limits_{i=1}^{n}F^{ij}e_ie_j\in C$, where $\{e_1,\dots,e_n\}$ is a basis
of the space $C$ and $\{F^{ij}\}$ is the matrix inverse with respect to $(e_i,e_j)_C$.

\vspace{1ex}

For the Frobenius pairs $(A,l^A)$ $(B,l^B)$ and for the linear operator $\phi: A\rightarrow B$ denote as
$\phi^*: B\rightarrow A$ the linear operator determined by the condition $(\phi^*(b),a)_A=(b,\phi(a))_A$.

\vspace{2ex}

\textit{\CF algebra} is a set $((A,l^A),(B,l^B),\phi)$ which consists of

1) a commutative Frobenius pair $(A,l^A)$;

2) an arbitrary Frobenius pair $(B,l^B)$;

3) a homomorphism of the algebras $\phi: A\rightarrow B$ such that the image $\phi(A)$ belongs to the center of $B$ and
$(\phi^*(b'),\phi^*(b''))_A =\tr K_{b'b''}$, where the operator $K_{b'b''}:B\rightarrow B$ is defined as $K_{b'b''}(b)=b'bb''$.

\vspace{1ex}

\textit{Equipped \CF algebra} is a set $((A,l^A),(B,l^B),\phi,U,\star)$ which consists of

1) a \CF algebra $((A,l^A),(B,l^B),\phi)$;

2) an involutive anti-automorphisms $\star :A\rightarrow A$ and $\star : B\rightarrow B$ such that
$l^A(x^\star)=l^A(x)$, $l^B(x^\star)=l^B(x)$, $\phi(x^\star)=\phi(x)^\star$;

3) an element $U\in A$ such that $U^2=K_A^\star$ and $\phi(U)=K_B^\star$.

The commutative Frobenius pairs are in one-to-one correspondence \cite{D} with
closed topological field theories in the meaning of \cite{At}. The \CF algebras
are in one-to-one correspondence \cite{AN} with open-closed topological field
theories in the meaning of \cite{Laz}, \cite{Moore},\cite{MS}. The equipped \CF algebras
are in one-to-one correspondence \cite{AN} with Klein topological field theories, i.e. with those
also defined on non-oriented Riemann surfaces \cite{AN}.

Each real representation of finite group generates a semi-simple equipped \CF
algebra \cite{LN}. Some of them describe the Hurwitz numbers of finite-sheeted coverings
 \cite{AN2}, \cite{AN3}. There exists a complete classification of the
semi-simple equipped \CF algebras \cite{AN}.

\vspace{2ex}

In the definitions above one needs to invert matrices. Therefore, one needs an additional
accuracy when dealing with the infinite-dimensional case. We additionally require
that the algebras can be presented as a direct (Cartesian) product of finite-dimensional
algebras $A=\prod\limits_{\gamma\in\mathfrak{C}}A_\gamma$, $B=\prod\limits_{\gamma\in\mathfrak{C}}B_\gamma$ and,
instead of functionals on $A$
and $B$, we consider the families of functionals $l^A=\{l^A_\gamma:A_\gamma\rightarrow\mathbb{C}\}$,
$l^B=\{l^B_\gamma:B_\gamma\rightarrow\mathbb{C}\}$ such that:

1)$(A_\gamma,l^A_\gamma)$ and $(B_\gamma,l^B_\gamma)$ are Frobenius pairs;

2)$\phi(A_\gamma)\in B_\gamma$ and restrictions $\phi_\gamma$ of the homomorphism  $\phi$ onto $A_\gamma$
gives rise to the \CF algebras $((A_\gamma,l^A_\gamma),(B_\gamma,l^B_\gamma),\phi_\gamma)$;

3) The involution $\star$ preserves the subalgebras $A_\gamma$, $B_\gamma$ and, along
with the projections $U_\gamma$ of the element $U\in A$ onto $A_\gamma$, gives rise to the equipped \CF
algebras $((A_\gamma,l^A_\gamma),(B_\gamma,l^B_\gamma),\phi_\gamma,U_\gamma,\star)$

\vspace{1ex}

\subsection{Cut-and-join operators}

Let us construct an equipped \CF algebra that consists of differential operators
of the space of functions of infinitely many variables $\{X_{ij}|i,j=1,\dots,\}$.

The role of algebra $A$ is played by the algebra $W$ of cut-and-join
operators $W(\Delta)$
\cite{MMN1},\cite{MMN2},\cite{AMMN},\cite{MMN3}.

We remind the construction of these latter. Consider differential operators of the form
$D_{ab}=\sum_{e=1}^\infty X_{ae}\frac{\partial}{\partial X_{be}}$ and
associate with the Young diagram $\Delta=[\mu_1,\mu_2,\dots,\mu_k]$ with
line lengths $\mu_1\geq\mu_2\geq\dots\geq\mu_k$ the numbers $m_j=m_j(\Delta)= |\{i|\mu_i=j\}|$ and
$\kappa(\Delta)= (|\Aut(\Delta)|)^{-1} =(\prod\limits_{j}m_j!j^{m_j})^{-1}$. Then one associates with the Young diagram
$\Delta$ \textit{cut-and-join operator} $W(\Delta)=\kappa(\Delta):\prod\limits_{j}(\tr D^j)^{m_j}:$,
where $D$ is the infinite-dimensional matrix with the elements
$D_{ab}=\sum_{e=1}^\infty X_{ae}\frac{\partial}{\partial X_{be}}$. We denote
the product of operators as $\circ$. Properties of these operators $W(\Delta)$
differs a lot from their finite-dimensional counterparts \cite{KO}.

Denote through $W_n$ the vector space generated by the cut-and-join operators of degree $n$.
The product $\circ$ of operators continued to infinite formal sums $w_1+w_2+w_3+\dots$ with $w_n\in W_n$
generates an associative commutative algebra $W$ of formal series of differential operators.

\subsection{Graph operators}
The cut-and-join operators are related (see \cite{MMN2}) with the Hurwitz numbers
of branching coverings generated by the compact oriented surfaces \cite{D}. The role of algebra $B$ is
played by the algebra $V$ of graph operators related with the Hurwitz numbers of branching coverings
generated by the foam surfaces \cite{AN3}.

We define \textit{bipartite graph of degree $n$} as a graph with $n$ edges and vertices
parted into two ordered groups $L=L(\Gamma)$ and $R=R(\Gamma)$, with the edges $E=E(\Gamma)$ connecting the
vertices from different groups. Isomorphisms are homomorphisms of graphs preserving the partition of
vertices into groups and ordering in the groups. Hereafter, we do not make any difference between isomorphic graphs.

We call the graph \textit{simple} if all its connected components are graphs with two vertices, and call the
graph obtained from $\Gamma$ by adding simple connected components \textit{the standard extension} of the graph
$\Gamma$. Denote through $\mathcal{E}_n(\Gamma)$ the set of all degree $n$ standard extensions of the graph $\Gamma$.
We put $\sigma_n(\Gamma)= \sum\limits_{\hat{\Gamma}\in\mathcal{E}_n(\Gamma)}
\frac{|\Aut(\hat{\Gamma})|}{|\Aut(\hat\Gamma\setminus\Gamma)|
|\Aut(\Gamma)|} \hat{\Gamma}$ at $n\geq |\Gamma|$ and $\sigma_n(\Gamma)=0$ at $n< |\Gamma|$,
see \cite{MMN4}.

\vspace{1ex}

Associate with the monomial $x=X_{a_1b_1}\dots X_{a_nb_n}$ of degree $n$ the bipartite graph $\Gamma(x)$
with the edges $\{E_1,\dots,E_m\}$ such that the edges $E_i$ and $E_j$ has a common left (right) vertex if
and only if $a_i=a_j$ ($b_i=b_j$). Now associate with the graph $\Gamma$ \textit{graph-variable}
$X_{\Gamma}=\frac{1}{|\Aut(\Gamma)|}\sum X$,
where the sum runs over all monomials $x$ such that $\Gamma(x)=\Gamma$. Denote through $X_n$ the
vector space generated by the graph-variables of degree $n$.

Associate with the operator $D= :D_{a_1b_1}\dots D_{a_nb_n}:$ the bipartite graph $\Gamma(\mathcal{D})$
with the edges $\{E_1,\dots,E_m\}$ such that the edges $E_i$ and $E_j$ has a common left (right) vertex if
and only if $a_i=a_j$ ($b_i=b_j$). Now associate with the graph $\Gamma$ the operator
$V[\Gamma]= \frac{1}{|\Aut(\Gamma)|}\sum\mathcal{D}$, where the sum runs over all operators
$\mathcal{D}$ such that $\Gamma(D)=\Gamma$.
We call such operators \textit{graph-operators}.

\vspace{1ex}

Let us give an action of the graph-operator of degree $n$ onto the graph-variables
of the same degree. The usual action of the graph-operator onto the graph-variable
results into a linear combination of the graph-variables with (generally) infinite
coefficients. Therefore, the correct definition of differentiation requires a
regularization. To construct it, consider, along with the (full) graph-operator and graph-variable
$V[\Gamma]$, $X_{[\Gamma']}$, the restricted graph-operator $V^N[\Gamma]$ and graph-variable $X_{[\Gamma']}^N$
defined similarly to the full ones, but the infinite set of variables $\{X_{ij}|i,j= 1,\dots,N\}$ being replaced
with the finite one $\{X_{ij}|i,j= 1,\dots,N\}$.

We define the action of the graph-operator $V^N[\Gamma]$ onto the graph-variable  $X_{[\Gamma']}^N$ as the usual
action of the differential operator multiplied with
$\frac{(N-|R(\Gamma)|)!}{N!}$. One can easily see that  $V^N[\Gamma](X^N_{[\Gamma']})$ is a linear
combination of the restricted graph-variables $X^N_{[\Gamma'']}$. Moreover, the coefficients of this linear
combination are the same at any $N>|E(\Gamma)|$.

We define $V[\Gamma](X_{[\Gamma']})= \lim\limits_{N\rightarrow\infty}
V^N[\Gamma](X^N_{[\Gamma']})$ at $|\Gamma|=|\Gamma'|$ and continue it naturally to
$V[\Gamma](X_{[\Gamma']})=0$ at $|\Gamma|>|\Gamma'|$ and $V[\Gamma](X_{[\Gamma']})=
V[\sigma_{|\Gamma'|}(\Gamma)](X_{[\Gamma']})$ at $|\Gamma|<|\Gamma'|$.

\vspace{1ex}

Denote through $V_n$ the vector space generated by the graph-operators of degree $n$ and through $V$ the vector
space of the formal differential operators $v_1+v_2+v_3+\dots$, where $v_n\in V_n$. Define on $V$ the operation
$\circ$, requiring that
the operator $V[\Gamma_1]\circ V[\Gamma_2]$ acts on all the graph-variables $X[\Gamma]$ as $V[\Gamma_1](V[\Gamma_2](X[\Gamma]))$.

\subsection{\CF algebra of differential operators}

The cut-and-join operators act on the space of graph-variables by
the usual differentiation. Define the homomorphism of algebras $f:
W\rightarrow V$ requiring that the operator $f(w)$ acts on all the
graph-variables in the same way as the operator $w$ (below we prove
that there exists such an operator).

The involution $D_{ab}\leftrightarrow D_{ab}$ gives rise to involutive automorphisms
$\star: W\rightarrow W$ and $\star: V\rightarrow V$.

The Young diagrams of degree $n$ are naturally identified with a basis of the center of
the group algebra of the symmetric group $S_n$. The sum $U_n$ of squares of all the elements of the
symmetric group that act without fixed points, belongs to the same center. This allows one to associate
an operator $r_n\in W_n$ with the element $U_n$. Now put $R=\sum\limits_{i=1}^\infty r_n$.

\vspace{1ex}

Our main result is the following

\begin{theorem}\label{t1.1} There exist decompositions
$W=\prod\limits_{\gamma\in\mathfrak{C}}W_\gamma$,
$V=\prod\limits_{\gamma\in\mathfrak{C}}V_\gamma$ and the family of functionals
$l^W=\{l^W_\gamma:W_\gamma\rightarrow\mathbb{C}\}$,
$l^V=\{l^V_\gamma:V_\gamma\rightarrow\mathbb{C}\}$ such that the set
$((W,l^W),(V,l^V),f, R,\star)$ forms an equipped \CF algebra.
\end{theorem}

The constructed algebra is a counterpart of the algebra of cut-and-join operators
for the foam-coverings (simplest examples of such operators are constructed in \cite{N}).
We discuss it in detail elsewhere.

\bigskip

\paragraph{{\bf Acknowledgements.}}
S.N. is grateful to MPIM and IHES for the kind hospitality and support.

Our work is partly supported by Ministry of Education and Science of
the Russian Federation under contract 8498,
by Russian Federation Government Grant No. 2010-220-01-077,
by NSh-3349.2012.2 (A.Mir. and A.Mor.) and 8462.2010.1 (S.N.),
by RFBR grants 10-02-00509 (A.Mir. and A.N.), 10-02-00499 (A.Mor.) and
by joint grants 11-02-90453-Ukr, 12-02-92108-Yaf-a,
11-01-92612-Royal Society.

\section{Finite-dimensional algebras of Young diagrams and bipartite graphs}

\subsection{Young diagrams}

First, we remind the standard facts that we will need. Denote through
$|\mathfrak{M}|$ the number of elements in the finite set $\mathfrak{M}$ and through
$S_n$ the symmetric group that acts by permutations on the set $\mathfrak{M}$ with $|\mathfrak{M}|=n$.
The permutation $\sigma\in S_n$ generates the subgroup
$<\sigma>$, with its action dividing $\mathfrak{M}$ ito the orbits $\mathfrak{M}_1,\dots,\mathfrak{M}_k$.
The set of numbers $|\mathfrak{M}_1|,\dots,|\mathfrak{M}_k|$ is called \textit{cyclic type of the permutation $\sigma$}.
It produces the Young diagram $\Delta(\sigma)= [|\mathfrak{M}_1|,\dots,|\mathfrak{M}_k|]$ of degree $n$. Permutations are
conjugated in $S_n$ if and only if they have the same cyclic type.

Linear combinations of permutations from $S_n$ form the group algebra
$G_n=G(S_n)$. We denote multiplication in this algebra as $"*"$.
Associate with each Young diagram $\Delta$ the sum $G_n(\Delta)\in G_n$
of all permutations of the cyclic type $\Delta$. These sums (which, for the sake of brevity, we denote with the same
symbol $\Delta$ as the corresponding Young diagram) form a basis of \textit{algebra of the conjugation classes}
$A_n\subset G_n$ coinciding
with the center of algebra $G_n$.

\begin{example} \label{ex1.0} $[1]*[1]=[1]$, $[2]*[2]=[11]$, $[11]*[11]=[11]$, $[3]*[3]=2[111]$.
\end{example}

The sum $U_n$ of squares of elements of the group $S_n$ belongs to the algebra $A_n$. Denote through
$l_n^A: A_n \rightarrow\mathbb{C}$ a linear functional equal to
$\frac{1}{n!}$ on the Young diagram with all lines having unit length and equal to zero otherwise. Denote
through $*: A_n\rightarrow A_n$ the identity map $a\mapsto a^*=a$.

\subsection{Bipartite graphs}
Describe following \cite{AN1},\cite{AN2} the operation of multiplication $*$ on the vector space $B_n$
generated by the set of class of isomorphism of the bipartite graphs of degree $n$.

Let $(L,E,R)$ and $(L',E',R')$ is a pair of the bipartite graphs with $n$ edges.
Denote $Hom(R,L')$ the set of maps $\chi:R\rightarrow L'$ which preserve the vertex valences. Every such map
is associated with the bipartite graph $(R,E_\chi,L')$ whose edges connect only the vertices $\tilde{v}$ and
$\chi(\tilde{v})$, where $\tilde{v}\in R$, the number of edges connecting $\tilde{v}$ and $\chi(\tilde{v})$
being equal to the valence of the vertex $\tilde{v}$.

We call a subset $F\subset E\times E'$ consistent with $\chi$, if the restrictions onto $F$ of the natural
projections $E\times E'\rightarrow
E$, $E\times E'\rightarrow E'$ are in one-to-one correspondence and
$\chi(R(e))=L(e')$ for any $(e,e')\in F$. Denote through
$M_\chi$ the set of all such $F$. Associate with the set $F\in
M_\chi$ the bipartite graph $(L,\overline{F},R')$ whose edges are the pairs of edges
$(e,e')\in F$ glued in the points $R(e)$ and $L(e')$. Denote through
$\Aut_F(L,\overline{F},R')\subset\Aut(L,\overline{F},R')$ the subgroup which consists of the automorphisms
inducing on the set $E$ the automorphism of the graph $(L,E,R)$.

Let us now construct the map $B_n\times B_n\rightarrow B_n$ by putting $[(L,E,R)][(L',E',R')]=$
$\sum_{\chi\in Hom(R,L')}\sum_{F\in
M_\chi}\frac{|\Aut((R,\overline{F},L'))|}
{|\Aut_F((L,\overline{F},R'))|}[(L,\overline{F},R')]$.
Continuing it by linearity, one obtains a binary operation which transforms $B_n$
into algebra.

This operation has a simple geometrical meaning. The contribution to the product
$[(L,E,R)]*[(L',E',R')]$ is given by the valence-preserving identifications of vertices
from $R$ and $L'$. As a result of such an identification, there emerge "a special graph" with vertices
$L\cup R'$ and edges that intersect on "the set of singularities"
$R=L'$. Product is called a linear combination of "resolutions" of these singularities, i.e. of (generating
the bipartite graph) pairwise gluing of edges from $(L,E,R)$ and $(L',E',R')$ coming to the common vertex.

\vspace{2ex}

The algebra $B_1$ looks as follows
$$
\left[\ododlo\,\right]*\left[\ododlo\,\right] = \left[\ododlo\,\right]
$$

The multiplication table for $B_2$ looks as
$$
\begin{array}{|c||c|c|c|c|}
\hline
&&&&\\
{*}& \left[\ododld\right]& \left[\oddva\right]& \left[\dvaod\right]
& \left[\dvadva\right]\\
&&&&\\
\hline\hline
&&&&\\
\left[\ododld\right]& \left[\ododld\right]& \left[\oddva\right]& 0& 0\\
&&&&\\
\hline
&&&&\\
\left[\oddva\right]& 0& 0& \left[\ododld\right]& \left[\oddva\right]\\
&&&&\\
\hline
&&&&\\
\left[\dvaod\right]& \left[\dvaod\right]& \left[\dvadva\right]& 0 & 0 \\
&&&&\\
\hline
&&&&\\
\left[\dvadva\right]& 0& 0& \left[\dvaod\right]& v\left[\dvadva\right]\\
&&&&\\
\hline
\end{array}
$$

\begin{example} The multiplication table for $B_3$ is

\vspace{1ex}

\centerline{\footnotesize{
$
\begin{array}{|c|c|c|c|c|c|c|c|c|c|c|}
\hline
&&&&&&&&&&\\
{*} & \left[\ododlt\right] & \left[\tritri\right] & \left[\odtri\right]
& \left[\triod\right] & \left[\dvaodlt\right] & \left[\oddvalt\right]
& \left[\tridva\right] & \left[\dvatri\right] & \left[\dvadvalt\right]
& \left[\zigzag\right] \\
&&&&&&&&&&\\
\hline
&&&&&&&&&&\\
\left[\ododlt\right]
&\left[\ododlt\right] &0&\left[\odtri\right]
&0 & 0& \left[\oddvalt\right]
&0&0&0&0 \\
&&&&&&&&&&\\
\hline
&&&&&&&&&&\\
\left[\tritri\right]
&0&\left[\tritri\right]&0
&\left[\triod\right]&0&0
&\left[\tridva\right]&0&0&0 \\
&&&&&&&&&&\\
\hline
&&&&&&&&&&\\
\left[\odtri\right]
&0&\left[\odtri\right]&0
&\left[\ododlt\right]&0&0
&\left[\oddvalt\right]&0&0&0 \\
&&&&&&&&&&\\
\hline
&&&&&&&&&&\\
\left[\triod\right]
&\left[\triod\right]&0&\left[\tritri\right]
&0&0&\left[\tridva\right]
&0&0&0&0 \\
&&&&&&&&&&\\
\hline
&&&&&&&&&&\\
\left[\dvaodlt\right]
&\left[\dvaodlt\right]&0&\left[\dvatri\right]
&0&0&\left[\dvadvalt\right]+ 
&0&0&0&0 \\
&&&&&&+ \left[\zigzag\right]&&&&\\
&&&&&&&&&&\\
\hline
&&&&&&&&&&\\
\left[\oddvalt\right]
&0&0&0
&0&3\left[\ododlt\right]&0&0
&3\left[\odtri\right]&\left[\oddvalt\right]&2\left[\oddvalt\right] \\
&&&&&&&&&&\\
\hline
&&&&&&&&&&\\
\left[\tridva\right]
&0&0&0
&0&3 \left[\triod\right]&0
&0&3\left[\tritri\right]&\left[\tridva\right]&2\left[\tridva\right] \\
&&&&&&&&&&\\
\hline
&&&&&&&&&&\\
\left[\dvatri\right]
&0&\left[\dvatri\right]&0
&\left[\dvaodlt\right]&0&0
&\left[\dvadvalt\right]+&0&0&0 \\
&&&&&&&+ \left[\zigzag\right]&&&\\
&&&&&&&&&&\\
\hline
&&&&&&&&&&\\
\left[\dvadvalt\right]
&0&0&0
&0&\left[\dvaodlt\right]&0
&0&\left[\dvatri\right]&\left[\dvadvalt\right]& \left[\zigzag\right]\\
&&&&&&&&&&\\
\hline
&&&&&&&&&&\\
\left[\zigzag\right]
&0&0&0
&0&2\left[\dvaodlt\right]&0
&0&2\left[\dvatri\right]&\left[\zigzag\right]
&2\left[\dvadvalt\right]+ \\
&&&&&&&&&&+ \left[\zigzag\right]\\
&&&&&&&&&&\\
\hline
\end{array}
$
}}
\end{example}

The sum $\mathfrak{e}_n^B= \sum\frac{\Gamma}{|\Aut(\Gamma)|}$ over the set of simple graphs $\Gamma$ of degree $n$
is the identity element of
algebra $B_n$. Denote through
$l_n^B:B\rightarrow\mathbb{C}$ the linear functional which is equal to
$\frac{1}{|\Aut(\Gamma)|}$ on the simple graphs $\Gamma$ and vanishes on all other graphs.

The involution $(L,E,R)\mapsto (L,E,R)^*=(R,E,L)$ induces an anti-isomorphism $*: B_n\rightarrow B_n$ (i.e. $(ab)^*=b^*a^*$).

\subsection{Homomorphisms of algebras}

Let $\mathfrak{N}$ be the set of all partitions of the set $\mathfrak{M}$ into non-empty subsets. Consider a vector
space $H_n$ over $\mathbb{C}$ with a basis $\mathfrak{N}$.

Construct following \cite{AN2} a homomorphism $\varrho: B_n\rightarrow\End(H_n)$. To this end, one suffices to
determine the image $\varrho(\Gamma)$ of the graph $\Gamma=(L,E,R)$. The vertices $L$ and $R$ give rise to
partitions of the set of edges $E$ into subsets $\sigma_L,\sigma_R$. The bijection $\chi:E\rightarrow\mathfrak{M}$
maps them to the partitions $\chi(\sigma_L), \chi(\sigma_R)\in\mathfrak{N}$. Denote through
$\varrho_\chi(\Gamma)\in\End(H_n)$ the endomorphism mapping to zero all the partitions from $\mathfrak{N}$ not equal
to $\chi(\sigma_R)$, and $\varrho_{\chi}(\chi(\sigma_R))=\chi(\sigma_L)$. Denote through $\varrho(\Gamma)\in\End(H_n)$
the sum of endomorphisms $\varrho_{\chi}$ over all bijections $\chi:E\rightarrow\mathfrak{M}$.
The map $\varrho:B_n\rightarrow\End(H_n)$ is an injective homomorphism of algebras.

The natural action of group $S_n$ on $\mathfrak{M}$ gives rise to a natural action of group $S_n$ on $\mathfrak{N}$ and,
hence, a homomorphism of the group algebra
$\tilde{\phi}_n:G_n\rightarrow\End(H_n)$. The image $\chi(B_n)\in\End(H_n)$ consists of all endomorphisms commuting with
$\tilde{\phi}_n(G_n)$ and, in particular, $\tilde{\phi}_n(A_n)\subset \chi(B_n)$. This allows one to define the
homomorphism $\phi_n=\varrho^{-1}\tilde{\phi}_n: A_n\rightarrow B_n$.

\begin{example}
$$\phi_1[1]=\left[\ododlo\right]$$
$$\phi_2[1,1]=\left[\dvadva\right]+\left[\ododld\right]$$
$$\phi_2[2]=\left[\dvadva\right]+\left[\ododld\right]$$
\end{example}

\begin{theorem} \label{t2.1} \cite{AN2} The set $((A_n,l^A_n),(B_n,l^B_n),\phi_n,U_n,\star)$ forms an equipped \CF algebra.
\end{theorem}

\vspace{2ex}

\section{Proof of the main theorem}

\subsection{Restricted multiplication}
First consider a few examples of the graph-operators and graph-variables and action of the graph-operators on the
graph-variables.

\begin{example} \label{ex2.3}
$$X_{\left[\ododlo\,\right]} = \sum_{a,b} X_{ab} \quad X_{\left[\dvadva\right]} =
\frac{1}{2}\sum_{\stackrel{a_1\neq a_2}{b_1\neq b_2}}^N
X_{a_1b_1}X_{a_2b_2}$$
$$X_{\left[\oddva\right]} =
\frac{1}{2}\sum_{\stackrel{a}{b\neq c}} X_{ab}X_{ac} \quad X_{\left[\dvaod\right]} =
\frac{1}{2}\sum_{\stackrel{a\neq b}{c}} X_{ac}X_{bc} \quad X_{\left[\ododld\right]} =\frac{1}{2}\sum_{a,b=1}X_{ab}^2$$
\end{example}

\begin{example} \label{ex2.2}
$$V\left[\ododlo\,\right] = \sum_{a,b} D_{ab} \quad
V\left[\dvadva\right] =
\frac{1}{2}\sum_{\stackrel{a_1\neq a_2}{b_1\neq b_2}}
: D_{a_1b_1}D_{a_2b_2}:$$
$$V\left[\oddva\right] =
\frac{1}{2}\sum_{\stackrel{a}{b\neq c}} : D_{ab}D_{ac}: \quad
V\left[\dvaod\right] =
\frac{1}{2}\sum_{\stackrel{a\neq b}{c}} : D_{ac}D_{bc}: \quad
V\left[\ododld\right] =\frac{1}{2}\sum_{a,b=1} : D_{ab}^2:$$
\end{example}

\begin{example} $$V\left[\ododlo\,\right](X_{\left[\ododlo\,\right]})=
\lim\limits_{N\rightarrow\infty}\frac{(N-1)!}{N!}\sum_{a,b}^N D_{ab} (\sum_{a',b'}^NX_{a'b'})= X_{\left[\ododlo\,\right]}$$
$$V\left[\dvadva\,\right](X_{\left[\dvadva\,\right]})= \lim\limits_{N\rightarrow\infty}\frac{(N-2)!}{N!}
\frac{1}{2}\sum^N_{\stackrel{a_1\neq a_2}{b_1\neq b_2}}
: D_{a_1b_1}D_{a_2b_2}: (\frac{1}{2}\sum_{\stackrel{a'_1\neq a'_2}{b'_1\neq b'_2}}^N
X_{a'_1b'_1}X_{a'_2b'_2})= X_{\left[\dvadva\,\right]}$$
$$V\left[\ododld\,\right](X_{\left[\ododld\,\right]})= \lim\limits_{N\rightarrow\infty}
\frac{(N-1)!}{N!}\frac{1}{2}\sum^N_{a,b=1}
: D_{ab}^2: (\frac{1}{2}\sum^N_{a',b'=1}X_{a'b'}^2)= X_{\left[\ododld\,\right]}$$
\end{example}

\textit{The restricted product} $\Gamma_1*\Gamma_2$ of the graph-operators $\Gamma_1$ and $\Gamma_2$ of degree $n$
is called the operator $pr_n(\Gamma_1 \Gamma_2)$, i.e. the projection onto the subspace of operators of degree $n$
the regularized product of differential operators $V[\Gamma_1]V[\Gamma_2]= \lim\limits_{N\rightarrow\infty}
V[\Gamma_1]^NV[\Gamma_2]^N$. This operation transforms $V_n$ into an associative algebra.

\begin{remark} When defining the product $*$ of graph-operators $V[\Gamma]$ and their action on the graph-variables,
one does not need any regularization if to restrict oneself to the graphs $\Gamma$ with the same number of the left
and right vertices. Then, the operator $V[\Gamma]$ is defined to act as the differential operator $\check{V}[\Gamma]=
\frac{1}{|\Aut(\Gamma)|}\sum\mathcal{D}$, where the sum runs over all
operators $\mathcal{D}= :D_{a_1b_1}\dots D_{a_nb_n}:$ such that $\Gamma(\mathcal{D})=\Gamma$, the sets of pairwise
distinct numbers among the sets $\{a_1,\dots,a_n\}$ and $\{b_1,\dots,b_n\}$ being coincident.
\end{remark}

Let us continue the correspondence between graphs and endomorphisms described in the previous subsection to the isomorphism
of the vector spaces $\theta_n: V_n\rightarrow \varrho(B_n)\in\End(H_n)$.

\begin{lemma}\label{l3.1} The map $\theta_n$ is an isomorphism of algebras, $V\left[\Gamma\right](X_{\left[\Gamma'\right]}) =
X_{\left[\Gamma*\Gamma'\right]}$ and $V\left[\Gamma\right]*V\left[\Gamma'\right] = V_{\left[\Gamma*\Gamma'\right]}$.
\end{lemma}

\begin{proof} Consider the monomial $D=:D_{a_1b_1}\dots D_{a_nb_n}:$, where $\Gamma(D)=\Gamma$ and the monomial
$X\left[\Gamma'\right]=X_{c_1d_1}\dots X_{c_nd_n}$, where $\Gamma(X)=\Gamma'$. Then, $D(X)$ is a linear combination
of all monomials of the form $X_{s_1d_1}\dots X_{s_nd_n}$ such that the set of numbers $\{s_1,\dots,s_n\}$ coincides
with the set of numbers $\{a_1,\dots,a_n\}$ and $\Gamma(X_{s_1d_1}\dots X_{s_nd_n})$ is one of the summands in
expansion of the product
$\Gamma*\Gamma'$. Thus, $V\left[\Gamma\right](X_{\left[\Gamma'\right]}) = X_{\left[\Gamma*\Gamma'\right]}$. Hence, the
homomorphicity of the map $\theta_n$ and $V\left[\Gamma\right]*V\left[\Gamma'\right] = V_{\left[\Gamma*\Gamma'\right]}$.
\end{proof}

Denote through $l^V_n: l^V_n\rightarrow\mathbb{C}$ the linear functional such that $l^V_n(V[\Gamma])= l^B_n(\Gamma)$.
From Lemma \ref{l3.1} follows

\begin{theorem}\label{t3.1} The correspondence $\Gamma\mapsto W[\Gamma]$ gives rise to an isomorphism of Frobenius
pairs $\mathcal{V}_n:(V_n,l^V_n)\rightarrow(B_n,l^B_n)$
\end{theorem}

\vspace{2ex}

\subsection{Algebra of poligraph-operators}

The direct product of algebras $V_i$ forms the algebra $V_{\uparrow}=\prod\limits_{i=1}^{\infty}V_i$, whose
elements we call  \textit{poligraph-operators}. Thus, poligraph-operator is an infinite sequence $v=(v_1,v_2,\dots)$,
where $v_i\in V_i$ with componentwise restricted product.

Denote through $X_n$ the linear space generated by the graph-variables of degree $n$.
\textit{Poligraph-variable} is an infinite sequence $x=(x_1,x_2,\dots)$, where $x_i\in X_i$. Denote through
$X_{\uparrow}= \prod\limits_{i=1}^{\infty}X_i$ the vector space generated by the poligraph-variables. Let action of the
poligraph-operator on
the poligraph-variable be given by the formula $v(x)=(v_1, v_2, \dots)(x_1, x_2, \dots)= (v_1(x_1), v_2(x_2), \dots)$.
Then $(v'*v'')=v'(v''(x))$.

The correspondence $V[\Gamma]\mapsto V[\rho_n(\Gamma)]$ gives rise to a homomorphism $\sigma_n:V_m\rightarrow V_n$ of vector
spaces. The set of these homomorphisms gives rise to a homomorphism of vector spaces $\sigma_\uparrow: V_m\rightarrow V_{\uparrow}$,
which, in its turn, gives rise to a homomorphism of vector spaces $\sigma_\uparrow: V\rightarrow V_{\uparrow}$.

\begin{theorem}\label{t3.2} The map  $\sigma_\uparrow: V\rightarrow V_{\uparrow}$ is an isomorphism of algebras.
\end{theorem}

\begin{proof} Immediately from the definitions it follows that $pr_n(\sigma_\uparrow (V[\Gamma]))= \sigma_n(V[\Gamma])$.
Let $x_n\in X_n$ and $x=\sigma_\uparrow(x_n)$ is the
graph-variable with all zero components but $n$-th, and the $n$-th component is equal to $x_n$. Then
$\sigma_{\uparrow}(V[\Gamma_1]\circ V[\Gamma_2]))(x)=\sigma_\uparrow((V[\Gamma_1]\circ V[\Gamma_2])(x_n))=
\sigma_\uparrow(V[\Gamma_1](V[\Gamma_2](x_n))) = \\ \sigma_\uparrow(V[\Gamma_1])(\sigma_\uparrow(V[\Gamma_2])(x)) =
(\sigma_\uparrow(V[\Gamma_1]))*(\sigma_\uparrow(V[\Gamma_2]))(x)$.
Thus, $\sigma_\uparrow (V[\Gamma_1]\circ V[\Gamma_2])) = \sigma_\uparrow (V[\Gamma_1])*\sigma_\uparrow (V[\Gamma_2])$.
Monomorphicity of the homomorphism $\sigma_\uparrow $ is immediate. Epimorphicity follows from theorem \ref{t3.1}.
\end{proof}

\begin{example} $$\sigma_\uparrow(V\left[\ododlo\,\right]\circ V\left[\ododlo\,\right])= \sigma_\uparrow(V\left[\ododlo\,\right]+
2V\left[\dvadva\right])= (V\left[\ododlo\,\right],4V\left[\dvadva\right],\dots)$$
$$\sigma_\uparrow(V\left[\ododlo\,\right])* \sigma_\uparrow(V\left[\ododlo\,\right])= (V\left[\ododlo\,\right],
2V\left[\dvadva\right],\dots)* (V\left[\ododlo\,\right], 2V\left[\dvadva\right],\dots)= $$
$$(V\left[\ododlo\,\right], 4V\left[\dvadva\right],\dots)$$
\end{example}

\vspace{2ex}

\subsection{Algebra of cut-and-join operators}

We start with a few examples of the cut-and-join operators and their products.

\begin{example} \label{ex1.1}
$D_e=W([1])=\sum\limits_{a\in\mathbb{N}} :D_{aa}:$ \quad
$W([11])=\frac{1}{2}\sum\limits_{a,b\in\mathbb{N}} :D_{aa}D_{bb}:\quad$
$W([2])=\frac{1}{2}\sum\limits_{a,b\in\mathbb{N}} :D_{ab}D_{ba}:\quad$
$W([21])=\frac{1}{2}\sum\limits_{a,b,c\in\mathbb{N}} :D_{ab}D_{ba}D_{cc}:\quad$
$W([3])=\frac{1}{2}\sum\limits_{a,b,c\in\mathbb{N}} :D_{ab}D_{bc}D_{ca}:$.
\end{example}

\begin{example} \label{ex1.2}

$W[1]\circ W[1] = \sum_{a,b=1}^N D_{aa}D_{bb}
= \sum_{a,b}^N :D_{aa}D_{bb}:\ + \sum_a D_{aa}= \\
2W[11]  + W[1]$

\vspace{1ex}

$W[1]\circ W[2] = W[2]\circ W[1] = 2W[2] + W[21],$

\vspace{1ex}

$W[1]\circ W[11] = W[11]\circ W[1] = 2W[11] + 3W[111] $

$W[2]\circ W[2] = W[11] + 3W[3] + 2W[22],$

\vspace{1ex}

$W[2]\circ W[11] = W[11]\circ W[2] = W[2] + 2W[21] + W[211],$

\vspace{1ex}

$W[11]\circ W[11] = W[11] + 6W[111] + 6W[1111],$
\end{example}

Denote through $W_n$ the vector space generated by differential operators of the form $W(\Delta)$,
where $|\Delta|=n$. Introduce on $W_n$ the structure of associative algebra defining \textit{restricted multiplication}
$*$ with the equality $w^1*w^2= pr_n(w^1\circ w^2))$ for $w^i\in W_n$. Denote through
$W_{\uparrow}=\prod\limits_{n=1}^{\infty}W_n$ the direct product of algebras $W_n$ with the restricted multiplication $*$.

\vspace{1ex}

Define a linear operator $\rho_n:W_m\rightarrow W_n$ assuming that $\rho_n(W)=0$ at $n< m$ and
$\rho_n(W[\Delta])= \frac{k!}{m!(k-m)!}:D_e^{(n-m)}W[\Delta]:$, where $k$ is the number of unit lines
in the Young diagram $\Delta$, at $n\geq m$. The correspondence
$W\mapsto\rho_n(W)$ gives rise to a homomorphism of vector spaces $\rho_n: W_m\rightarrow W_n$. The set of
homomorphisms $\rho_n:W_m\rightarrow W_n$ gives rise to a homomorphism of vector spaces
$\rho_\uparrow: W_m\rightarrow W_{\uparrow}$. Continue the operator $\rho_\uparrow: W_m\rightarrow W_{\uparrow}$
up to the linear operator $\rho_\uparrow: W\rightarrow W_{\uparrow}$.

\begin{example} \label{ex1.5}
$\rho_\uparrow(W[1])=(W[1],2W[11],3W[111],\dots)$,

\vspace{1ex}

$\rho_\uparrow(W[2])= (W[2],W[21],W[211],\dots)$,

\vspace{1ex}

$\rho_\uparrow(W[11])=(W[11],3W[111],4W[1111],\dots)$,

\vspace{1ex}

$\rho_\uparrow(W[3])=(W[3],W[31],W[311],\dots)$,

\vspace{1ex}

$\rho_\uparrow(W[21])= (W[21],2W[211],3W[2111],\dots)$,

\vspace{1ex}

$\rho_\uparrow(W[111])= (W[111],4W[1111],5W[1111],\dots)$.
\end{example}

\begin{theorem}\label{t3.3} The map $\rho_\uparrow: W\rightarrow W_{\uparrow}$ is an isomorphism of algebras.
\end{theorem}

\begin{proof} Immediately from the definitions it follows that $pr_n(\rho_\uparrow (W(\Delta)))= \rho_n(W(\Delta))$ and
$pr_n(\rho_\uparrow (W(\Delta_1)\circ W(\Delta_2)))=$ $pr_n(\rho_\uparrow (\rho_n(W(\Delta_1))\circ \rho_n(W(\Delta_2))))=$
$pr_n(\rho_\uparrow (\rho_n(W(\Delta_1))*\rho_n(\Delta_2)))=
\rho_n(W(\Delta_1))*\rho_n(W(\Delta_2))= pr_n(\rho_\uparrow(W(\Delta_1))*\rho_\uparrow(W(\Delta_2)))$.
Thus, $\rho_\uparrow (W(\Delta_1)\circ W(\Delta_2))= \rho_\uparrow(W(\Delta_1))*\rho_\uparrow(W(\Delta_2))$.
Monomorphicity of the homomorphism $\rho_\uparrow $ is immediate. Epimorphicity follows from the equality
$W_{\uparrow}= \sum\limits_{n=1}^{\infty}\rho_{\uparrow}(W_{n})$.
\end{proof}
\begin{example} \label{ex1.6}
$$\rho_\uparrow (W([1]\circ W([1]))= \rho_\uparrow (W([1] + 2W([11])) =
\rho_\uparrow (W([1]) + 2\rho_\uparrow ( W([11]))=$$
$$=(W[1],2W[11],3W[111],\dots) + 2(W[11],3W[111],4W[1111],\dots) =
(W[1],4W[11],9W[111],\dots)$$.
On the other hand, $\rho_\uparrow (W([1])* \rho_\uparrow (W([1]))=\rho_\uparrow(W([1])) * \rho_\uparrow(W([1]))= \\
(W[1],2W[11],3W[111],\dots) * (W[1],2W[11],3W[111],\dots) = (W[1],4W[11],9W[111],\dots)$
\end{example}

\begin{example} \label{ex1.7}
$$\rho_\uparrow (W([2]\circ W([2]))= \rho_\uparrow ( W[11] + 3W[3] + 2W[22]) =$$
$$=\rho_\uparrow (W[11]) + 3\rho_\uparrow (W[3]) + 2\rho_\uparrow (W[22])=$$
$$=(W[11],3W[111],6W[1111],\dots) + 3(W[3],W[31],W[311],\dots) + 2(W[22],W[221],W[2211],\dots) =$$
$$=(W[11],3(W[3]+W[111]),3W[31]+2W[22]+6W[1111],\dots)$$.
On the other hand, $\rho_\uparrow (W([2])* \rho_\uparrow (W([2]))=\rho_\uparrow^W(W[2]) * \rho_\uparrow^W(W[2])=\\
=(W[2],W[21],W[211],\dots) * (W[2],W[21],W[211],\dots) =\\= (W[2]*W[2],W[21]*W[21],W[211]*W[211],\dots)=$

\vspace{1ex}

$=(W[11],3(W[3]+W[111]),3W[31]+2W[22]+6W[1111],\dots)$
\end{example}

Denote through $l^W_n: l^W_n\rightarrow\mathbb{C}$ the linear functional such that $l^W_n(W[\Delta])= l^A_n(\Delta)$.
Then, as proved in \cite{MMN2} (Lemma 4.4),

\begin{theorem}\label{t3.4} The correspondence $\Delta\mapsto W[\Delta]$ gives rise to an isomorphism
of the Frobenius pairs $\mathcal{W}_n:(A_n,l^A_n)\rightarrow(W_n,l^W_n)$
\end{theorem}

\vspace{2ex}

\subsection{\CF structures}

Define a linear operator $f_n: W_n\rightarrow V_n$ with the equality $f_n(w)=pr_n(f(w))$.
Their set gives rise to a linear operator $f_{\uparrow}: W_{\uparrow}\rightarrow V_{\uparrow}$.

These are a few examples of action of the cut-and-join operators on the graph-variables.

\begin{example} $$W[1](X_{\left[\ododlo\right]})=
\sum_{a} D_{aa}(\sum_{a,b}X_{ab})=X_{\left[\ododlo\right]}$$
$$W[1](X_{\left[\dvadva\right]})= \frac{1}{2} W[1,1](X_{\left[\dvadva\right]})=
\frac{1}{4}:\sum_{a}D_{aa}\sum_{b}D_{bb} :(\frac{1}{2}\sum_{\stackrel{a_1\neq a_2}{b_1\neq b_2}}
X_{a_1b_1}X_{a_2b_2})=
\frac{1}{2}X_{\left[\dvadva\right]}$$
$$W[1](X_{\left[\ododld\right]})=\frac{1}{2} W[1,1](X_{\left[\ododld\right]})=
\frac{1}{4}:\sum_{a}D_{aa}\sum_{b}D_{bb} :(\frac{1}{2}\sum_{a,b}
X_{ab}^2) =\frac{1}{2}X_{\left[\ododld\right]}$$
$$W[2](X_{\left[\dvadva\right]})=
\frac{1}{2}:\sum_{a,b}D_{ab}D_{ba}:(\frac{1}{2}\sum_{\stackrel{a_1\neq a_2}{b_1\neq b_2}}
X_{a_1b_1}X_{a_2b_2}) =X_{\left[\dvadva\right]}$$
$$W[2](X_{\left[\ododld\right]})=
\frac{1}{2}:\sum_{a,b}D_{ab}D_{ba}:(\frac{1}{2}\sum_{a,b}
X_{ab}^2) =X_{\left[\ododld\right]}$$
\end{example}

\begin{lemma} \label{l3.2} $\rho_\uparrow f=f_\uparrow\rho_\uparrow$
\end{lemma}

\begin{proof} Consider the operator $w=W[\Delta]$  and the graph-variable $x$ of degree $n$. Then,
$(f_\uparrow(\rho_\uparrow(w))(x)= (f_n(w)(x))_{\uparrow}=(\rho_\uparrow f(w))(x)$,
where $(f_n(w)(x))_{\uparrow}$ is the poligraph-variable with the single non-zero component $f_n(w)(x)\in X_n$
\end{proof}

\begin{lemma} \label{l3.3} If $n=|\Delta|$, then $f_n(W[\Delta])= V[\phi_n(\Delta)]$
\end{lemma}

\begin{proof} The operator  $W[\Delta]$ can be represented in the form
$W[\Delta]= \kappa(\Delta)\sum\limits_{a_1,...,a_n\in\mathbb{N}}
:D_{a_1a_{\sigma(1)}}\cdots D_{a_na_{\sigma(n)}}:$, where $\sigma\in S_n$ is a permutation of the cyclic type $\Delta$.
In accordance with the definition, the operator $V[\phi_n(\Delta)]$ has the same form.
\end{proof}

\begin{theorem}\label{t1.5} The set $((W,l^W),(V,l^V),f, R,\star)$ forms an equipped \CF algebra.
\end{theorem}

\begin{proof} Due to Theorems \ref{t3.2}, \ref{t3.3} and Lemma \ref{l3.2}, one suffices to prove that the isomorphisms
$\mathcal{W}_n$ and  $\mathcal{V}_n$ give rise to isomorphisms of \CF algebras
$((A_n,l^A_n),(B_n,l^B_n),\phi_n, U_n,\star)$ and $((W_n,l^W_n),(V_n,l^V_n),f_n, R_n,\star)$. Homomorphism of Frobenius pairs
is proved in Theorems \ref{t3.1}, \ref{t3.4}. The relation $f_n\mathcal{W}_n=\mathcal{V}_n\phi_n$
follows from Lemma \ref{l3.3}. The remaining requirements follow immediately from
the definitions.
\end{proof}

\end{document}